\documentclass[11pt]{article}
\usepackage[T1]{fontenc}
\usepackage{lmodern} 
\usepackage[utf8]{inputenc}
\usepackage{amsfonts,amsmath,amssymb}
\usepackage[a4paper]{geometry}

\begin{document}

\title{\bf A New Modular Division Algorithm \\
and Applications}
\author{M. S. Sedjelmaci \hspace{.7cm} C. Lavault \\
{\small \sl LIPN CNRS UPRES-A 7030, Universit\'e Paris-Nord, 93430 Villetaneuse, France} \\
{\small E-mail: \texttt{$\{$lavault,sms$\}$@lipn.univ-paris13.fr}}
}
\date{\empty} 
\maketitle

\vspace{-.8cm}
\begin{abstract}
The present paper proposes a new parallel algorithm for the modular 
division $u/v\bmod \beta^s$, where $u,\; v,\; \beta$ and $s$ are 
positive integers $(\beta\ge 2)$. The algorithm combines the classical 
add-and-shift multiplication scheme with a new propagation carry 
technique. This ``Pen and Paper Inverse'' ({\em PPI}) algorithm, 
is better suited for systolic parallelization in a ``least-significant 
digit first'' pipelined manner. Although it is equivalent to Jebelean's 
modular division algorithm~\cite{jeb2} in terms of performance 
(time complexity, work, efficiency), the linear parallelization 
of the {\em PPI} algorithm improves on the latter when the input size 
is large. The parallelized versions of the {\em PPI} algorithm leads 
to various applications, such as the exact division and the digit 
modulus operation (dmod) of two long integers. It is also applied 
to the determination of the periods of rational numbers as well 
as their $p$-adic expansion in any radix $\beta\ge 2$.

\end{abstract}

\newtheorem{thm}{Theorem}[section] 
\newtheorem{defi}{Definition}[section] 
\newtheorem{lem}{Lemma}[section]
\newtheorem{cor}{Corollary}[section] 
\newtheorem{prop}{Proposition}[section]
\newenvironment{proof}{%
				    \begin{trivlist}{\setlength{\itemsep}{0cm}
					  \item[]{\bf Proof}}%
				   }{\hfill $\square$
				    \end{trivlist}}
\newenvironment{rem}{%
                      \begin{trivlist}{\setlength{\itemsep}{0cm}
                         \item[]{\bf Remark:}}%
                    }{%
                      \end{trivlist}}
\newcommand{\divi}{\mbox{div}}

\bibliographystyle{article}
\def\bibfmta#1#2#3#4{{#1} {#2}, {\em #3}, #4}
\bibliographystyle{book}
\def\bibfmtb#1#2#3#4{{#1} {\em #2}, {#3}, #4}

\section{Introduction}
The modular division of two positive integers $u$ and $v$ modulo a radix 
$\beta\ge 2$ to the power $s$ is defined as $(u/v)\bmod \beta^s$. 
It is used in many topics in theoretical computer science, such as $p$-adic 
computation, cryptography, Computer Algebra systems, etc. When $s$ is small, 
say less than a machine word (i.e., 8, 16 or 32 bits), the modular division 
is completed by several algorithms in $O(s^2)$ time complexity. Such is the 
case for the {\em Extended Euclidean algorithm} ({\em EEA}) in~\cite{knu}. 
However, the {\em EEA} is not efficient for large $s$, because it entails 
long quotient-remainder divisions of two long positive integers at each step.

In~\cite{jeb2}, Jebelean proposes a more efficient algorithm: {\em Modiv}. 
It is better suited for systolic parallelization in a {\em LSF} 
({\em ``least-significant digit first''}) pipelined manner. Actually, 
experiments performed in~\cite{jeb2} using SACLIB Computer Algebra system 
indicate more than 20 times speed-up in computing the inverse modulo a 
power of $\beta=2^{29}$. The speed-up was measured in comparison with the 
{\em EEA} algorithm improved by Lehmer~\cite{leh} and Collins~\cite{col}.

Jebelean also presents {\em Ediv}, a new algorithm for exact division of two 
long integers which improves on the classical quotient-remainder algorithm 
described in~\cite{knu}, when it is known in advance that the remainder 
is zero. Although it is presented as an application to {\em Ediv}, 
{\em Modiv} may rather be considered as the true main algorithm. As a matter 
of fact, {\em Ediv} is obtained easily by running {\em Modiv} with the 
specific parameter $s=\ell_\beta(u)-\ell_\beta(v)+1$ (see 
Subsection~\ref{not} for the notation $\ell_\beta(u)$). Hence, we rather focus on modular division 
algorithms throughout the paper.

\medskip In Section 2, we compare and discuss sequential and parallel versions of 
{\em Modiv} and of the {\em ``Pen and Paper Inverse''} ({\em PPI}) algorithm 
for modular division. Section 3 is devoted to the parallelization of the 
{\em PPI} algorithm, where a linear parallel {\em PPI} ({\em ParPPI}) 
algorithm using a new carry propagation technique is presented. Section 4, 
gives applications of the {\em ParPPI} algorithm to the exact division 
and the {\em digit modulus}' (dmod) operations on two long integers, as well 
as to the determination of the periods of rational numbers and their 
$p$-adic expansion in any given radix $\beta\ge 2$. Section 5 concludes 
with some remarks.

\section{Modular Division Algorithms}
Let $x\equiv (u/v)\pmod {\beta^{s}}$. {\em Modiv} follows the process below. 
After shifting the operands $u$ and $v$ until they become relatively prime 
to $\beta$, $x_{0}$ (the least-significant digit of $x$) is found from 
$u_{0}$ and $v_{0}$ (the least-significant digits of $u$ and $v$) in the 
form $x_0\equiv u_0v_0^{-1}\pmod \beta$.

Next, once $x_0$ is obtained, the next coefficient $x_1$ is performed by 
applying again the same process to $(u-x_0v)/\beta$ and $v$.

\subsection{Notation and Example} \label{not}
Throughout the paper we assume that any integer is expressed in radix $\beta\ge 2$, 
where $\beta$ and $v$ are coprime. As in~\cite{web}, $\ell_\beta(u)$ denotes the number 
of digits needed to represent $u$ in radix $\beta$ and $\gcd(a,b)$ denotes 
the greatest commun divisor of integers $a$ and $b$.

\medskip The algorithm designed in the paper is based on a simple technique: 
the classical ``Pen and Paper'' multiplication, which also works
in {\em LSF} ({\em ``least-significant digit first''}) manner. 
Roughly speaking, the algorithmic scheme consists in guessing and 
filling the missed digits of $x$ from right to left, so that the puzzle 
adapts the classical multiplication, and the process stops when the $s$-th 
digit of $x$ is found. This is the reason why the algorithm is called the 
{\em ``Pen and Paper Inverse''} algorithm, or {\em PPI} for short.

In Table~\ref{tab:ppi}, a simple example describes this inverse multiplication scheme 
for $u=37,229$, $v=1,543$ and $\beta=2$ and $s=7$.

\begin{table}[t]
\begin{center}
\footnotesize
\begin{tabular}{|l|l|}
\hline
 & \\
$\quad \;\cdots \;\cdots\; \cdots$ 1 $\;\gets x$ unknown & $\quad \!$ 1 1 0 1 0 1 1 $\;\gets x=107$ \\
$\times$ 0 0 0 0 1 1 1 $\;\gets v\bmod 2^7$ & $\times$ 0 0 0 0 1 1 1 $\;\gets v\bmod 2^7$ \\
$\cdots \;\cdots\; \cdots\; \cdots \!$ 1 & $\quad \!$ 1 1 0 1 0 1 1 \\
$\cdots \quad \cdots \quad \cdots$ & $\quad \!$ 1 0 1 0 1 1 \\
$\cdots \quad \cdots \quad \cdots$ & $\quad \!$ 0 1 0 1 1 \\
= 1 1 0 1 1 0 1 $\;\gets u\bmod 2^7$ & = 1 1 0 1 1 0 1 $\;\gets u\bmod 2^7$ \\
& \\
\hline
 & \\
Before the {\em PPI} algorithm & After the {\em PPI} algorithm \\
 & \\
\hline
\end{tabular}
\caption{The {\em PPI} computation of $(37,229/1,543)\bmod 2^7$}
\label{tab:ppi}
\end{center}
\end{table}

\section{The Algorithms {\em Modiv} and {\em PPI}}
Let us recall the sequential and the parallel versions of the algorithm 
{\em Modiv} proposed in~\cite{jeb2}.

\subsection{The sequential version {\em SeqModiv}}

\bigskip \noindent \hrulefill\, {\bf \em Sequential Algorithm SeqModiv}\, \hrulefill

\smallskip \noindent {\em Input:}\ $u,\; v>0$, with $\gcd(v,\beta)=1$, and $s$ a positive 
integer. \\
{\em Output:}\ $(u/v)\bmod \beta^s$.

\smallskip $v_0$ := $v\bmod \beta$ ; $a$ := $v_0^{-1}\bmod$\ \hspace{3cm}  /*~initialization~*/

\smallskip {\bf for} $k$ := $0$ {\bf to} $s-1$ {\bf do} 

\hspace*{1cm} $x_k$ := $au\bmod \beta$ ;

\hspace*{1cm} $u$ := $(u-x_kv\bmod \beta^{s-k})/\beta$

{\bf endfor}

\smallskip \noindent {\bf return}\ $x=(x_{s-1},\ldots,x_0)$

\smallskip \noindent \hrulefill{}\hrulefill

\medskip The time complexity of {\em SeqModiv} is $O(s^2)$. In order to compare the 
parallel version of {\em Modiv} with the parallelization of the {\em PPI} 
algorithm, we first recall {\em ParModiv}, the parallel version of 
{\em Modiv}. In {\em ParModiv}, a carry save technique is used to handle 
the carry propagation (see~\cite{jeb3,mul}).

\medskip For every integer $x$ such that $|x|<\beta^2$, the notation $(a,b)$ := $x$ 
is defined as
$$(a,b)\stackrel{def}{=} (x \; \divi\ \beta,x\bmod \beta)\ \mbox{with}\ 0\le 
b<\beta,$$
where $\divi$ is the ``integral division'': 
$x \; \divi\ \beta \stackrel{def}{=} \lfloor x/\beta\rfloor$.

\subsection{The Parallel version {\em ParModiv}}

\bigskip \noindent \hrulefill\, {\bf \em Parallel Algorithm ParModiv}\, \hrulefill

\medskip \noindent {\em Input:}\ $u,\; v>0$, with $\gcd(v,\beta)=1$, and $s$ a positive 
integer. \\
{\em Output:}\ $(u/v)\bmod \beta^s$.

\smallskip {\bf for} $i$ := $1$ {\bf to} $s$ {\bf pardo} $y_i$ := $0$\ {\bf endfor}\ \hspace{2cm} /*~initialization~*/

\hspace*{.5cm} $a$ := $v_0^{-1}\bmod \beta$

\smallskip {\bf for} $k$ := $0$ {\bf to} $s-2$ {\bf do}

\hspace*{.5cm} $x_k$ := $au_k\bmod \beta$ ;

\hspace*{.5cm} $y_{k+1}$ := $y_{k+1}-(x_kv_0) \; \divi \beta$\ \hspace{3cm} /*~update of the first carry~*/

\smallskip \hspace*{1cm} {\bf for} $i$ := $1$ {\bf to} $s-k-1$ {\bf pardo} 

\hspace*{1.5cm} $(y_{k+i+1},u_{k+i})$ := $u_{k+i}-x_kv_i+y_{k+i}$

\hspace*{1cm} {\bf endfor}

\smallskip {\bf endfor}

\smallskip \hspace*{1cm} $x_{s-1}$ := $au_{s-1}\bmod \beta$

\smallskip \noindent {\bf return}\ $x=(x_{s-1},\ldots,x_1,x_0)$

\smallskip \noindent \hrulefill{}\hrulefill

\subsection{The PPI Algorithm}
Assume $u$ and $\beta$ are coprime. Otherwise, $u=\beta^q u$' for some 
positive integers $q$ and $u$', such that $u$' and $\beta$ are coprime. 
Then, $x_i=0$ for $i=0,1,\ldots,(q-1)$. Since $(u/v)\bmod \beta^s =\beta^q \big((u'/v)\bmod \beta^{s-q}\big)$, 
the algorithm variables reduces to the parameters $u$' and $s'=s-q$.

In the algorithm the $s$ least-significant digits of $u$ only are needed; 
and thus we may use $u\bmod \beta^s$ instead of $u$. Set

$$u\bmod \beta^s =\sum_{i=0}^{s-1} u_i \beta^i\ \quad \mbox{and}\ \quad v\bmod \beta^s = \sum_{j=0}^{s-1} v_j \beta^j,$$
and if $\ell_\beta(u)<s$, let $u_i=0$ for all $i\ge \ell_\beta(u)$.

On the above assumptions, the algorithm described in Table~\ref{tab:ppi} expresses as 

\bigskip \noindent \hrulefill\, {\bf \em Sequential PPI Algorithm}\, \hrulefill

\medskip \noindent {\em Input:}\ $u,\; v>0$, with $\gcd(u,\beta)=\gcd(v,\beta)=1$, and $s$ a positive integer. \\
{\em Output:}\ $(u/v) \bmod \beta^s$.

\medskip $a$ := $v_0^{-1}\bmod \beta$ ; $x_0$ := $(au_0)\bmod \beta$ ;\ \hspace{2cm} /*~initialization~*/

\smallskip $c_1$ := $(x_0v_0) \; \divi\ \beta$

\newpage
{\bf for} $k$ := $1$ {\bf to} $s-1$ {\bf do}\ \hspace{5cm} /*~loop~*/

\hspace*{1cm} $L_k$ := $\sum_{j=1}^{k} v_j x_{k-j}+c_k$ ;

\hspace*{1cm} $x_k$ := $a(u_k-L_k)\bmod \beta$ ; 
$c_{k+1}$ := $(L_k+x_kv_0) \; \divi\ \beta$ 

{\bf endfor}

\smallskip \noindent {\bf return}\ $x=(x_{s-1},\ldots,x_0)$

\smallskip \noindent \hrulefill{}\hrulefill

\bigskip \noindent {\bf Remarks:}
\begin{itemize}
\item In the case when $\beta=2$, the algorithm turns out to be simpler.

\item Notice that the loop is very similar to a triangular linear system of the form
$$A\,X\; =\; B,$$
with solution $X=\; ^t\!(x_{s-1},\ldots,x_1)$, where $B=\;^t\!(b_{s-1},\ldots,b_1)$ 
is a given vector and $A=(a_{i,j})$ $(1\le i,\; j\le s-1)$, is defined as
$$a_{i,j}\; \stackrel{def}{=}\;
\left\{ \begin{array}{lll}
	v_1 & \mbox{if}\ i=j,\\
	v_{i-j+1} & \mbox{if}\ i>j,\\
	0 & \mbox{if}\ i<j.
\end{array} \right.$$
The parallelization process described in Section~\ref{lpv} applies to the above triangular system.

\item In place of $s$, the constant $m=\ell_\beta(v\bmod \beta^s)$ can be used for updating $L_k$, 
because it is the current number of digits of $v$ needed in the computation (see Fig.~1). 
The constant $m$ allows faster computations: $m\le s$ and, when $m<s$, all the useless 
computations corresponding to $v_m,\ldots,v_{s-1}$ are eliminated. Hence, $L_k$ is updated as follows:
$$L_k\ \mbox{:=}\ \sum_{1\le j\le \mu} v_j x_{k-j}+c_k,\ \quad \mbox{where}\ \; \mu=\min(k-1,m-1).$$

\item For any given $k$, the carry $c_{k+1}$ satifies the relation $L_k + x_kv_0 = u_k + \beta c_{k+1}$.

Moreover, the worst-case time complexity of the above algorithm occurs 
for all pairs $(u,v)$ such that $(u+v)\equiv 0\pmod {\beta^{s}}$, with 
output $x=(u/v)\bmod \beta^s=\beta^s - 1$.
In that case, the largest value of $L_k$ is $\beta k - 1$. Therefore, 
for all $k$,
$$L_k\le \beta k-1\le \beta s - \beta - 1.$$
\end{itemize}

\section{Linear Parallelized {\em PPI} Algorithms} \label{lpv}
In the {\em PPI} algorithm, the output $x$ is obtained step by step, 
least-significant digit first. The digits $a$, $u_k$ and the 
least-significant digit of $L_k$, namely $L_k\bmod \beta$, are only 
needed to compute $x_k$. However, $L_k\bmod \beta$ is obtained 
after computing all of the two-digits products $v_j\times x_{k-j}$ 
and their sum.

This variant significantly increases the running time of the algorithm 
and prevents any efficient systolic implementation. In order to overcome 
the difficulty, we make use of the followings two facts:
\begin{enumerate}
\item Every two-bits product $v_j\times x_{k-j}$ can be computed and added to $L_k$ as soon as $x_{k-j}$ is found.

\item The goal of the parallelization is to break the computation of the sums 
$\displaystyle{\sum_{1\le j\le k-1} v_j x_{k-j}}$. As far as the carry propagation is concerned for updating 
the $L_i$'s (for $i=k+1,\ldots,s-1$), the carry-save technique used in {\em ParModiv} can be applied successfully 
to the parallelization of the {\em PPI} algorithm {\em ParPPI}.
\end{enumerate}

\subsection{A Linear Parallel PPI Algorithm ParPPI}\label{lpppi}

\bigskip \noindent \hrulefill\, {\bf \em ParPPI Algorithm (version~1)}\, \hrulefill

\medskip \noindent {\em Input:}\ $u,\; v>0$, such that $\gcd(u,\beta)=\gcd(v,\beta)=1$,
and $s$ a positive integer. \\
{\em Output:}\ $(u/v)\bmod \beta^s$. 

\medskip {\bf for} $i$ := $0$ {\bf to} $s$ {\bf pardo} $(y_i,L_i)$ := 0 
{\bf endfor} ;\ \hspace{1cm} /*~initialization~*/

$a$ := $v_0^{-1}\bmod \beta$

\smallskip \hspace*{.5cm} {\bf for} $k$ := $0$ {\bf to} $s-2$ {\bf do}\ 
\hspace{3cm} /*~main loop~*/

\hspace*{1cm} $x_k$ := $a(u_k-L_k-y_k)\bmod \beta$

\smallskip \hspace*{1.5cm} {\bf for} $i$ := $0$ {\bf to} $s-k-1$ {\bf pardo}

\hspace*{2cm} $(y_{k+i+1},L_{k+i})$ := $L_{k+i} + x_kv_i + y_{k+i}$

\hspace*{1.5cm} {\bf endfor}

\smallskip \hspace*{.5cm} {\bf endfor}

\smallskip \hspace*{1cm} $x_{s-1}$ := $a(u_{s-1}-L_{s-1}-y_{s-1})\bmod \beta$

\smallskip \noindent {\bf return}\ $x=(x_{s-1},\ldots,x_1,x_0)$

\smallskip \noindent \hrulefill{}\hrulefill

\subsubsection{Comparison with {\em ParModiv}}

{\em ParPPI} (version~1) and {\em ParModiv} are linear in terms of {\em ``surface''} 
(i.e., the maximal number of processors needed in the algorithms) and time complexity. 
The algorithms are equivalent, but they also slightly differ in the following points:
\begin{itemize}
\item Every variable in the {\em ParPPI} algorithm consists 
of a single non-negative digit, whereas {\em ParModiv} uses a signed 
double digit for the variables $y_j$.

\item In contrast with {\em ParModiv}, the {\em ParPPI} algorithm 
does not change the input $u$.

\item In the {\em ParPPI} algorithm, all carries $y_{k+i}$ are updated in 
parallel. By contrast the first carry is updated serially in {\em ParModiv}.
\end{itemize}
Although they are not important regarding the design of the algorithms 
themselves, the above differences cause substantial time improvements 
when $s$ is large, and these algorithms are intensively used to devise 
efficient GCD algorithms, for example. (See~\cite{jeb3,sor,web}.)
 
\subsection{A New Carry Propagation Technique}
We now describe a new carry propagation technique, which propagates 
the carries alternately. This new technique, called {\em ``alternated 
carry''}, is illustrated in the second version of {\em ParPPI}, where 
the main loop of {\em ParPPI} (Version~1) designed in Subsection~\ref{lpppi} 
can be rewritten as follows:

\bigskip \noindent \hrulefill\, {\bf \em Main Loop of the ParPPI Algorithm (version~2)}\, \hrulefill

\medskip {\bf for} $k$ := $0$ {\bf to} $s-2$ {\bf do}

\hspace*{.5cm} $x_k$ := $a(u_k-L_k)\bmod \beta$

\smallskip \hspace*{1cm} {\bf for} $\;i$ := $0$ {\bf to} $s-k-1$ {\bf pardo} 
$L_{k+i}$ := $L_{k+i} + x_k v_i$\ \,{\bf endfor}

\hspace*{1cm} {\bf for} $n$ := $0$ {\bf to} $\lfloor (s-k-1)/2\rfloor$ 
{\bf pardo}

\hspace*{1.5cm} $L_{k+2n+1}$ := $L_{k+2n+1} + L_{k+2n} \; \divi\ \beta$ ; 
$L_{k+2n}$ := $L_{k+2n}\bmod \beta$

\hspace*{1cm} {\bf endfor}

\smallskip {\bf endfor}

\smallskip \noindent \hrulefill{}\hrulefill

\bigskip Theorem~\ref{bound} below shows that, whatever the input size $s$, all the $L_i$'s are bounded. 
Therefore, the {\em ParPPI} algorithm (version~2) is also linear in terms of surface 
(i.e., the maximal number of processors needed) and time complexity.

\begin{thm} \label{bound}
For any $\beta\ge 2$, $s>1$, $k\le s-1$ and $n\le \lfloor (s-k-1)/2\rfloor$, 

($i)$~$L_{k+2n}\le \beta - 1$.

$(ii)$~$L_{k+2n+1}\le \beta^2 + \beta - 2$.
\end{thm}
\begin{proof} By induction on $k$.

\smallskip {\bf Basis:}\ Obviously, $L_{2n}=L_{2n+1}=0$. So, $L_{2n}\le \beta - 1$ 
and $L_{2n+1}\le \beta^2 + \beta - 2$.

\medskip {\bf Induction step:}\ for any $i\ge k$, let $L_i(k)$ denote the value of 
$L_i$ at the end of the $k$-th iteration. Suppose that inequalities $(i)$ and $(ii)$ 
hold for a given positive integer $k$. After the computation of $x_{k+1}$, 
and since $x_{k+1},\; v_j\le \beta - 1$ for any $k$ and $j$, we have
\begin{align} 
L_{k+1+2n+1}(k) + x_{k+1} v_{2n+1} &\;\le \;\beta^2 -\beta,\ 
\qquad \quad \;\mbox{for}\ \; n\le \lfloor (s-k-2)/2\rfloor, \label{eq1} \\
L_{k+1+2n}(k) + x_{k+1} v_{2n} &\;\le \;2\beta^2 -\beta-1,\ 
\quad \mbox{for}\ \; n\le \lfloor (s-k-1)/2\rfloor. \label{eq2}
\end{align}
At the end of the $(k+1)$-st iteration Eq.~\eqref{eq1} and Eq.~\eqref{eq2} yield
\begin{align*}
L_{k+2n+1}(k+1) &\;\le \;(2\beta^2-\beta-1)\bmod \beta\le \beta - 1,\ \qquad \text{and}\\
L_{k+2n+2}(k+1) &\;\le \;\beta^2-\beta+(2\beta^2-\beta-1) \; \divi\ \beta = \beta^2 + \beta - 2.
\end{align*}
\end{proof}
\begin{rem}
The alternated carry propagation technique uses 3 digits for $L_{k+2n+1}$ 
and one digit for $L_{k+2n}$. Thus, only $2s$ digits are needed in the algorithm.
\end{rem}

\section{Applications}
The {\em ParPPI} algorithm can also be used in several applications. 
All these algorithms use the same ``pen and paper'' multiplication 
technique combined with a carry propagation technique (either Version~1 
or Version~2). Since it is new, we rather use the second version in the 
applications.

\subsection{Exact Division Algorithms}
Algorithms for exact division compute the exact quotient u/v of two long 
integers u and v, when it is known in advance that the remainder is zero. 
The exact division is easily completed by using the {\em ParPPI} algorithm 
with the parameter $r=\ell_\beta(u)-\ell_\beta(v)+1$ in place of $s$. 
Then, $r$ is no longer the ``input size'' but rather expresses the 
difference between the sizes of inputs $u$ and $v$~\cite{jeb2}.

\subsection{The dmod Operation}
Let $\ell_\beta(u)=s$ and $\ell_\beta(v)=t$, with $s\ge t$. The dmod operation is defined as
$$\mbox{dmod}_\beta(u,v) \stackrel{def}{=} |xv-u|/\beta^r,\ \mbox{where}\ 
r=s-t+1\ \mbox{and}\ x\equiv (u/v)\!\!\!\pmod{ \beta^{r}}.$$

The algorithm below is similar to the {\em ParPPI} algorithm. It simply multiplies 
$x_k$ and $v$, and simultaneously substracts $u$ from $x_k v$. The subtraction 
is performed by $\beta$-complement. The $\beta$-complement of $u_k$ is defined as follows:
$$u'_k\; \stackrel{def}{=}\;
\left\{ \begin{array}{lll}
	(\beta-u_k) & \mbox{if}\ k=r, \\
	(\beta-1-u_k) & \mbox{if}\ k>r.
\end{array}
\right.$$

We assume that $xv$ and $u$ are two $(s+1)$ digits numbers; so we set $u_s=0$.

\bigskip \noindent \hrulefill\, {\bf \em ParPPI {\em dmod} Algorithm}\, \hrulefill

\smallskip \noindent {\em Input:}\ $u,\; v$ two positive numbers, $\ell_\beta(u)=s$, 
$\ell_\beta(v)=t$, with $s\ge t$, and $\gcd(v,\beta)=1$. \\
{\em Output:}\ $x\equiv (u/v)\pmod {\beta^{r}}$ and dmod$_\beta(u,v) = |xv-u|/\beta^r$, where $r=s-t+1$.

\medskip {\bf for} $i$ := $0$ {\bf to} $s+1$ {\bf pardo} $L_i$ := 0 {\bf endfor} ;\ \hspace{2cm} /*~initialization~*/

$a$ := $v_0^{-1}\bmod \beta$ ; $r$ := $s-t+1$

\smallskip \hspace*{.5cm} {\bf for} $k$ := $0$ {\bf to} $r-1$ {\bf do}

\hspace*{.8cm} $x_k$ := $a(u_k-L_k)\bmod \beta$

\smallskip \hspace*{1cm} {\bf for} $\;i$ := $0$ {\bf to} $t-1$ {\bf pardo} 
$L_{k+i}$ := $L_{k+i} + x_kv_i$\ {\bf endfor}

\hspace*{1cm} {\bf for} $n$ := $0$ {\bf to} $\lfloor (s-k-1)/2\rfloor$ 
{\bf pardo}

\hspace*{1.5cm} $L_{k+2n+1}$ := $L_{k+2n+1} + L_{k+2n}\; \divi\ \beta$ ;

\hspace*{1.5cm} $L_{k+2n}$ := $L_{k+2n}\bmod \beta$

\hspace*{1cm} {\bf endfor}

\smallskip \hspace*{.5cm} {\bf endfor}

\smallskip \hspace*{1cm} {\bf for} $\;i$ := $0$ {\bf to} $s-r$ {\bf pardo} 
$L_{r+i}$ := $L_{r+i} + u'_{r+i}$\ {\bf endfor}

\smallskip \hspace*{.5cm} {\bf for} $k$ := $r$ {\bf to} $s$ {\bf do}

\smallskip \hspace*{1cm} {\bf for} $n$ := $0$ {\bf to} $\lfloor (s-k)/2\rfloor$ 
{\bf pardo}

\hspace*{1.3cm} $L_{k+2n+1}$ := $L_{k+2n+1} + L_{k+2n} \; \divi\ \beta$ ;

\hspace*{1.3cm} $L_{k+2n}$ := $L_{k+2n}\bmod \beta$

\hspace*{1cm} {\bf endfor}

\smallskip \hspace*{.5cm} {\bf endfor}

\smallskip \hspace*{1.3cm} $w$ := $(L_s,\ldots,L_r)$

\smallskip \hspace*{1.3cm} {\bf if} $L_{s+1}=0$ {\bf then} $w$ := $\beta^t - w$

\smallskip \noindent {\bf return}\ $x,\;w$

\smallskip \noindent \hrulefill{}\hrulefill

\begin{rem} If $L_{s+1}=0$, $xv-u = w-\beta^t < 0$; and $xv-u$ is easily given by 
$\beta$-complement. Moreover the algorithm provides the sign of $xv-u$.
Note also that the algorithm computes the modular division and the digit modulus operations simultaneously.
\end{rem}

\subsection{A Linear Surface-Time Multiplication}
The same algorithm also applies to perform a multiplication which is 
also linear in terms of surface and time complexity. However, the 
{\em ParPPI} multiplication algorithm is not even efficient, since 
the best sequential multiplication algorithms are $O(s\log^c s)$, for some constant $c>0$.

\newpage
\noindent \hrulefill\, {\bf \em The ParPPI Multiplication Algorithm}\, \hrulefill

\smallskip \bigskip \noindent {\em Input:}\ $u,\; v$ two positive numbers, $\ell_\beta(u)=s$, 
$\ell_\beta(v)=t$. \\
{\em Output:}\ $L=uv$.

\medskip {\bf for} $i$ := $0$ {\bf to} $s+t$ {\bf pardo} $L_i$ := 0 {\bf endfor}

\smallskip {\bf for} $k$ := $0$ {\bf to} $t-1$ {\bf do}

\smallskip \hspace*{.5cm} {\bf for} $\;i$ := $0$ {\bf to} $s-1$ {\bf pardo} 
$L_{k+i}$ := $L_{k+i} + v_ku_i$\ {\bf endfor}

\hspace*{.5cm} {\bf for} $n$ := $0$ {\bf to} $\lfloor (s-1)/2\rfloor$ 
{\bf pardo} 

\hspace*{1cm} $L_{k+2n+1}$ := $L_{k+2n+1} + L_{k+2n} \; \divi\ \beta$ ;

\hspace*{1cm} $L_{k+2n}$ := $L_{k+2n}\bmod \beta$

\hspace*{.5cm} {\bf endfor}

{\bf endfor}

\smallskip {\bf for} $k$ := $t$ {\bf to} $s+t-1$ {\bf do}

\smallskip \hspace*{.5cm} {\bf for} $n$ := $0$ {\bf to} $\lfloor (s+t-k-1)/2\rfloor$ {\bf pardo}

\hspace*{1cm} $L_{k+2n+1}$ := $L_{k+2n+1} + L_{k+2n} \; \divi\ \beta$ ;

\hspace*{1cm} $L_{k+2n}$ := $L_{k+2n}\bmod \beta$

\hspace*{.5cm} {\bf endfor}

\smallskip {\bf endfor}

\smallskip \noindent {\bf return}\ $L=(L_{s+t-1},\ldots,L_1,L_0)$

\smallskip \noindent \hrulefill{}\hrulefill

\subsection{The $p$-adic Expansion of Rationals}
The ordered sequence of the $s$ digits of $x \equiv (u/v)\pmod {\beta^{s}}$ 
represents exactly the expected Hensel code H$(u,v;\beta^s)$. The Hensel 
code is thus directly given by the {\em ParPPI} algorithm. Note that 
the $p$-adic expansion is obtained step by step, when $s$ tends to infinity.

\subsection{Computation of Periods of Rational Numbers}
Let $u/v$ be a rational number, with $u<v$, and let $v$ and $\beta$ be coprime. 
Let $T$ be the periodic part of the expansion, called the period of $u/v$, 
and let $t$ be the length of $T$ in base $\beta$. Then, $u/v=T/(\beta^t-1)$, 
and thus $vT\equiv (-u)\pmod {\beta^{t}}$.

Now using the {\em PPI} algorithm with the parameters $-u$, $v$ and $t$ yields $T\equiv -u/v\pmod \beta^t$.

\smallskip Note that $T$ is obtained in a ``least-significant digit first'' manner, 
whereas the classical division provides the period in a ``most-significant digit first'' manner.

\section{Conclusion}
The new modular division {\em PPI} algorithms (sequential and parallel 
variants) enjoy various interesting properties.

It is straightforward to derive many {\em LSF} algorithms from the 
{\em ParPPI} algorithm. Such is the case for the exact division, 
for the $p$-adic expansion and the periods of rational numbers, 
and for the multiplication and the dmod operations as well. All above 
applications are founded on the one and same identical scheme. This 
makes it easier to construct interactions between them, and thus provide 
an homogeneous collection of routines which may be extended later.

The {\em ``alternated carry propagation''} is a new carry propagation 
technique, which is also fruitful for add-and-shift algorithms.

\smallskip The parallel algorithms proposed herein follow the same classical
multiplication scheme along with a carry propagation technique (either 
carry save or alternated carry). These algorithms are all linear ($O(s)$) 
in terms of surface $S(s)$ (i.e., the maximal number of processors needed) 
and time complexity $T(s)$, where $s$ is the input size. As a consequence, 
the work (or cost), $W(s)=S(s)\times T(s)$, is $O(s^2)$, and they are neither  
in $\cal NC$, nor efficient (recall that the best sequential algorithms 
solving the same problems are $O(s\log^c(s))$, for some constant $c>0$). 
Note also that their parallel speed-up is $O(\log^c s)$, with efficiency 
$O(\log^c(s)/s)$.

{\em ParModiv} and the main {\em ParPPI} algorithm are equivalent, since 
{\em ParModiv} is also surface-time linear (and hence optimal). However, our 
algorithm may significantly improve when used intensively within several 
efficient GCD algorithms, for example~\cite{jeb3,leh,sor,web}. A new algorithm 
is also presented, which performs the modular division as well as the digit 
modulus of two long positive integers simultaneously. 

\smallskip All the {\em ParPPI} algorithms described are suitable for systolic 
implementation in a ``least-significant digit first'' manner, because 
all decisions in the procedures are taken by using the lower digits of the 
operands. Hence they can be well aggregated to other systolic algorithms 
in the arithmetic of multiprecision rational numbers. {\em LSF} processing 
is also used in the most efficient systolic algorithms for multiprecision 
rational arithmetic. Among them, one may mention long integer 
multiplication~\cite{atr} and addition/substraction~\cite{knu} algorithms, 
the Brent-Kung systolic GCD algorithm~\cite{brk} and the algorithms 
in~\cite{jeb1,jeb2,jeb3}.

\medskip The present work continues and complements our previous 
investigations~\cite{sel} in improving the algorithm for modular division. 
The combined effects of these improvements allow several basic routines in 
Computer Algebra systems to run more efficiently.


\begin{thebibliography}{99}

\bibitem{atr}\bibfmta
{A.J.~Atrubin}{A one-dimensional iteration multiplier}
{IEEE Trans. on Computers}{C-14, 1965, 394-399}

\bibitem{brk}\bibfmta
{R.P.~Brent, H.T.~Kung}{Systolic VLSI arrays for linear-time GCD computation}
{in VLSI'83}{Anceau and Aas eds., 1983, 145-154}

\bibitem{col}\bibfmtb
{G.E.~Collins}{Lecture note on arithmetic algorithms}{Un. of Wisconsin}
{1980}

\bibitem{jeb1}\bibfmta
{T.~Jebelean}{Systolic Algorithms for Exact Division}{RISC-Linz Report}
{92-71, Dec.~1992}

\bibitem{jeb2}\bibfmta
{T.~Jebelean}{A Generalization of the Binary GCD Algorithm}{in Proc. of the 
International Symposium on Symbolic and Algebraic Computation (ISSAC'93)}
{1993, 111-116}

\bibitem{jeb3}\bibfmta
{T.~Jebelean}{An Algorithm for Exact Division}{Journal of Symbolic 
Computation}
{15, 1993, 169-180}

\bibitem{knu}\bibfmtb
{D.E. Knuth}{The art of computer programming: seminumerical 
algorithms}{Vol.~2, 2nd ed}
{Addisson Wesley, 1981}

\bibitem{leh}\bibfmta
{D.H.~Lehmer}{Euclid's algorithm for large numbers}{American Math. Monthly}
{45, 1938, 227-233}

\bibitem{mul}\bibfmtb
{J.M.~Muller}{Arithm\'etiques des ordinateurs}{Masson}{1989}

\bibitem{sel}\bibfmta
{M.S.~Sedjelmaci, C.~Lavault}{Improvements on the accelerated integer GCD 
algorithm}
{Information Processing Letters}{61, 1997, 31-36}

\bibitem{sor}\bibfmta
{J.~Sorenson}{Two Fast GCD Algorithms}{J. of Algorithms}{16, 1994, 110-144}

\bibitem{swa}\bibfmtb
{Earl~E.~Swartlander~Jr}{Computer Arithmetic (tutorial)}
{Vol.~1, IEEE Computer Society Press}{1990}

\bibitem{web}\bibfmta
{K.~Weber}{Parallel implementation of the accelerated integer GCD algorithm}
{J. of symbolic Computation (Special Issue on Parallel Symbolic Computation)}
{21, 1996, 457-466}

\end{thebibliography}
\end{document}